\documentclass[submission,copyright,creativecommons]{eptcs}
\providecommand{\event}{16th International Conference on Quantum Physics and Logic (QPL 2019)}

\usepackage{amsmath}
\usepackage{amssymb}
\usepackage{amsthm} 
\usepackage{bbm}
\usepackage{mathrsfs}
\usepackage[all]{xy}
\xyoption{2cell}
\UseAllTwocells

\theoremstyle{plain}

\newtheorem{lemma}{Lemma}[section]
\newtheorem{proposition}{Proposition}[section]
\newtheorem{theorem}{Theorem}[section]
\theoremstyle{definition}
\newtheorem{definition}{Definition}[section]

\newcommand{\As}{\mathscr{A}}
\newcommand{\Bs}{\mathscr{B}}
\newcommand{\B}{\mathbf{B}}
\newcommand{\id}{\mathrm{id}}
\newcommand{\Set}{\mathbf{Set}}
\newcommand{\RStr}{\mathbf{RStr}}
\newcommand{\Pred}{\mathrm{Pred}}
\newcommand{\Stat}{\mathrm{Stat}}

\title{Effectus of Quantum Probability on Relational Structures}
\author{
 Octavio Zapata \
\institute{Department of Computer Science\\University College London}
\email{ocbzapata@gmail.com}
}

\def\titlerunning{Effectus of Quantum Probability on Relational Structures}
\def\authorrunning{O. Zapata}

\begin{document}
\maketitle

\begin{abstract}
 The notion of effectus  from categorical logic is relevant in the emerging field of categorical probability theory. In this context, stochastic maps are represented by maps in the Kleisli category of some probability monad. Quantum homomorphisms from combinatorics and quantum information theory are the Kleisli maps of certain sort of quantum monad. We show that the Kleisli category of this quantum monad is an effectus. This gives rise to notions of quantum validity and conditioning. 
\end{abstract}

\section{Introduction}
A \emph{graph} $G$ consists of a set of vertices $V(G)$ and a set of edges $E(G)\subseteq V(G)\times V(G)$. 
By definition $E(G)$ is a binary relation on the vertex set $V(G)$. We write $v\sim v'$ to denote a pair of adjacent vertices $v,v'\in V(G)$, \textit{i.e.} a pair $(v,v')\in V(G)\times V(G)$ in the edge relation $(v,v')\in E(G)$.  
A  \emph{graph homomorphism} $G \to H$ is given by a function $f\colon V(G) \to V(H)$ between vertices preserving edges: 
\[v \sim v' ~\text{ in }~G\quad\Rightarrow\quad f(v) \sim f(v')~\text{ in }~ H\]
 %Maps of graphs are also known as \emph{graph homomorphisms}.

Consider the following game involving a given pair of graphs $G$ and $H$, played by Alice and Bob playing against a Verifier. Their goal is to establish the existence of a graph homomorphism from $G$ to $ H$. The game is `non-local' which means that Alice and Bob are not allowed to communicate during the game, however they are allowed to agree on a strategy before the game has started.  In each round Verifier sends to Alice and Bob vertices $v_1, v_2 \in V(G)$, respectively; in response they produce outputs $w_1, w_2 \in V(H)$. They win the round if the following conditions hold:
\[v_1 = v_2 \Rightarrow w_1 = w_2 \qquad \text{ and } \qquad v_1 \sim v_2 \Rightarrow w_1 \sim w_2\] 

If there is indeed a graph homomorphism $G \to H$, then Alice and Bob can win any round of the game described above by using such homomorphism as strategy for responding accordingly. Conversely, they can win any round with certainty only when there is a graph homomorphism $G \to H$. A strategy for Alice and Bob in which they win with probability 1 is called a \emph{perfect strategy}. Hence, the existence of a perfect strategy is equivalent to the existence of a graph homomorphism.

 In cases where no classical homomorphism exists, one can use quantum resources in the form of a maximally entangled bipartite state, where Alice and Bob can each perform measurements on their part, to construct perfect strategies. These strategies are called quantum because they use quantum resources.

 We write $M_d(\mathbb{C})$ for the set of all $d\times d$ matrices with complex entries ($d\geq 1$), and $\mathbbm{1}\in M_d(\mathbb{C})$ for the $d\times d$ identity matrix. Let $E\in M_n(\mathbb{C})$ and $F\in M_m(\mathbb{C})$ be two complex square matrices, $n,m\geq1$. Their \emph{tensor product}  is the matrix defined as $E\otimes F:=( e_{ij}F)\in M_{nm}(\mathbb{C})$ if $E=(e_{ij})$ with $i,j=1,\dots,n$. 

\begin{definition} A \emph{quantum perfect strategy} for the homomorphism game from $G$ to $H$ consists of a complex unitary vector $\psi \in \mathbb{C}^{d_A}\otimes \mathbb{C}^{d_B}$ for some $d_A, d_B\geq 1$ finite, and families $( E_{vw} )_{w \in V(H)}$ and $( F_{vw} )_{w \in V(H)}$ of $d_A\times d_A$ and $d_B\times d_B$ complex matrices for all $v\in V(G)$, satisfying:
\begin{itemize}
\item[(1)] $\sum_{w\in V(H)} E_{vw} = \mathbbm{1} \in M_{d_A}(\mathbb{C})$ and  $\sum_{w\in V(H)} F_{vw} = \mathbbm{1} \in M_{d_B}(\mathbb{C})$;
\item[(2)] $w\neq w' \quad \Rightarrow \quad \psi^{\ast} (E_{vw}\otimes F_{vw'})\psi = 0$;
\item[(3)] $v \sim v' \land w\not\sim w'\quad \Rightarrow \quad\psi^{\ast} (E_{vw}\otimes F_{v'w'})\psi = 0$.
\end{itemize}
\end{definition}

Observe that the definition of quantum perfect strategies forgets  the two-person aspect of the game and shared state, leaving a matrix-valued relation as the witness for existence of a quantum perfect strategy. Recall perfect strategies are in bijection with graph homomorphisms.  This gives rise to 
the notion of quantum graph homomorphism. This concept was introduced in \cite{roberson2013variations}, as a generalisation of the notion of quantum chromatic number from \cite{cameron2007quantum}. Analogous results for constraint systems are proved in \cite{cleve2014characterization, manvcinska2016quantum, abramsky2017quantum, atserias2018quantum}. 

\begin{definition} \label{def:qgraph} A \emph{quantum graph homomorphism} from $G$ to $H$ is an indexed family 
$( E_{vw} )_{v \in V(G), w \in V(H)}$ of $d\times d$ complex matrices, $E_{vw} \in M_d(\mathbb{C})$, for some $d\geq 1$, such that:
\begin{itemize}
\item[(1)] $E^{\ast}_{vw}=E^{2}_{vw}=E_{vw}$ for all $v\in V(G)$ and $w\in V(H)$;
\item[(2)] $\sum_{w\in V(H)} E_{vw} = \mathbbm{1} \in M_d(\mathbb{C})$ for all $v\in V(G)$;
\item[(3)] $(v=v'\land w \neq w') \lor( v \sim v' \land w \not\sim w')\quad \Rightarrow \quad E_{vw} E_{v'w'} = 0$.
\end{itemize}
\end{definition}

An important further step is taken in~\cite{roberson2013variations}: a construction $G \mapsto \mathsf{M} G$ on graphs is introduced, such that the existence of a quantum graph homomorphism from $G$ to $H$ is equivalent to the existence of a graph homomorphism of type $G \to  \mathsf{M} H$. This construction is called the \emph{measurement graph}, and it turns out to be a graded monad on the category of graphs. The Kleisli morphisms of this monad are exactly the quantum homorphism between graphs of \cite{roberson2013variations, manvcinska2016quantum, atserias2018quantum}. One can show equivalence between these three different notions: quantum homomorphisms, quantum perfect strategies,  and certain kind of (classical) homomorphisms between graphs \cite{abramsky2017quantum}.

 Monads are used in formal semantics of functional and probabilistic programming languages. Building on the work of Giry~\cite{giry1982categorical}, and inspired by algebraic methods in program semantics, the study of various `probability' monads has evolved and became part of a new branch of categorical logic called effectus theory. The main goal of effectus theory is to describe the salient aspects of quantum computation and logic using the language of category theory. This description includes probabilistic and classical logic and computation as special cases~\cite{cho2015introduction, jacobs2015new}. Quantum perfect strategies form a category.  In this paper we shall see that the category of  quantum perfect strategies, or quantum graph homomorphisms, is an effectus (see Theorem~\ref{thm:effKlQ}).  However, specialisation of Theorem~\ref{thm:effKlQ} to the case where the underlying category is the category simple undirected graphs (which is the context used in $\cite{roberson2013variations, manvcinska2016quantum}$) is impossible as this category does not have a terminal object. This is why we introduce the quantum monad in the context of the more general notion of relational structures rather than that of simple graphs.

\section{Preliminaries}
\label{pre}
 Disjoint union is the coproduct in the category $\Set$ of sets and functions. 
 The
\emph{disjoint union} of two sets $X,Y$ is the defined to be the set  \[X+Y:=\{(x,1): x\in X\}\cup\{(y,2): y\in Y\}.\] 
One can define two functions $X\overset{\kappa_1}{\rightarrow}X+Y\overset{\kappa_2}{\leftarrow} Y$, as $\kappa_1(x):=(x,1)$ and $\kappa_2:=(y,2)$, for all $x\in X$ and $y\in Y$, respectively. The functions $\kappa_1, \kappa_2$ are called \emph{coprojections}. 
For any pair of functions $X\overset{p}{\rightarrow} Z \overset{q}{\leftarrow} Y$, the function $[p,q]\colon X+Y\to Z$ called \emph{cotupling} is given by:
\[
[p,q](v):=
\begin{cases}
p(v)	&~v\in X \\
q(v)	&~v\in Y
\end{cases}\qquad (v\in X+Y)
\]
 If  $f\colon A \to B,~g\colon X\to Y$ are functions, then the function $f + g\colon A+X \to B+Y$ is defined as \[f+g:=[\kappa_1 \circ f,~\kappa_2 \circ g].\]

 The empty set $0:=\emptyset$ is the initial object of $\Set$, and any choice of a singleton set $1:=\{\ast\}$ is terminal in $\Set$. The unique function $!_X\colon X\to 1$ is given by $x\mapsto \ast$ for each $x\in X$. Hence, the category $\Set$ has finite coproducts $(+,0)$ and a terminal object $1$.
 
 \begin{definition} \label{def:eff}
An \emph{effectus} is category $\mathbf{B}$ with finite coproducts $(+,0)$ and a terminal object $1$, such that for all $X,Y$ objects of $\B$, the following commutative squares are pullbacks:
\[\label{eff:pbcks}\xymatrix@C=30pt@R=20pt{
X + Y \ar[r]^{!_X + \id_Y} \ar[d]_{\id_X + !_Y} & 1 + Y\ar[d]^{\id_1 + !_Y} && X \ar[r]^{!_X} \ar[d]_{\kappa_1}& 1\ar[d]^{\kappa_1} \\ 
X + 1 \ar[r]_{!_X + \id_1} & 1 +1 && X+Y\ar[r]_{!_X + !_Y}&1+1} \]
and the following maps in $\B$ are jointly monic:
\[\label{eff:joint}\xymatrix@C=80pt@R=20pt{
(1 + 1) + 1 \ar@/^/[r]^-{\gamma_1:=[[\kappa_1,\kappa_2], \kappa_2]}\ar@/_/[r]_-{\gamma_2:=[[\kappa_2,\kappa_1], \kappa_2]} & 1 + 1
}\]
Joint monicity of  $\gamma_1, \gamma_2$ means that given maps $f,g\colon X\to (1+1)+1$, we have: \[\gamma_1 \circ f= \gamma_1 \circ g~\wedge~\gamma_2 \circ f= \gamma_2 \circ g \quad \Rightarrow \quad f = g\].
\end{definition}

The category $\Set$ is the effectus used for modelling classical (deterministic, Boolean) computations. The following result is well-known (see, \emph{e.g.}~\cite[Example 4.7]{jacobs2015new}). 

\begin{theorem}\label{prop:set-eff}
 The category $\Set$ is an effectus. 
\end{theorem}
\begin{proof}
We know how pullbacks are constructed in $\Set$. For the first pullback condition 
from Definition~\ref{def:eff}, let $P$ be the set of pairs $(x,y)\in (X+1)\times (1+Y)$ such that $(!_X+\id_1)(x)=(\id_1+!_Y) (y)$. 
Note that we have: \[(X+1)\times(1+Y) \cong (X\times 1) + (1 \times 1) + (X \times Y) + (1 \times Y)\] Let $X+1=X+\{1\}$ and $1+Y= \{0\}+Y$. 
 By cases: 
 \begin{itemize}
 \item[(1)] $(x,y)\in X\times \{0\}$ implies $(x,y)=(x,0)$, and so $(!_X+\id_1)(x)= 0=(\id_1+!_Y)(0)$ for all $x\in X$, thus $X\times 1 \subseteq P$; 
 \item[(2)] $(x,y) \in \{1\}\times \{0\}$ implies $(!_X+\id_1)(1)=1\neq 0= (\id_1+!_Y) (0)$, so $1\times1 \not\subseteq P$; 
 \item[(3)] $(x,y)\in X\times Y$ implies $(!_X+\id_1)(x)\neq (\id_1+!_Y) (y)$, so $X\times Y \not\subseteq P$; 
 \item[(4)] $(x,y)\in \{1\}\times Y$ implies $(x,y)=(1,y)$, and so $(!_X+\id_1)(1)= 1=(\id_1+!_Y) (y)$ for all $y\in Y$, thus $1\times Y \subseteq P$. 
 \end{itemize}
 Hence, the pullback $P$ is indeed given by $(X\times 1)+(1 \times Y) \cong X+Y$.

For the second pullback condition %of~(\ref{eff:pbcks}) 
from Definition~\ref{def:eff}, take $1=\{0\}$ and consider the set of pairs $(w, 0)\in (X+Y) \times 1$ such that 
$(!_X+!_Y)(w)=\kappa_1(0)$.  Note that $(X+Y) \times 1\cong (X \times 1) + (Y \times 1)$. 
By cases: 
 \begin{itemize}
 \item[(1)] if $(w,0)\in X\times 1$ then $(!_X+!_Y)(w)=0=\kappa_1(0)$ for all $w\in X$;  \item[(2)] if $(w,0)\in Y\times 1$ then $(!_X+!_Y)(w)=1\neq 0 = \kappa_1(0)$ for all $w\in Y$. 
 \end{itemize}
 Thus the pullback is indeed given by $X\times 1 \cong X$.

For the joint monicity requirement from Definition~\ref{def:eff}, we consider sets $1+1+1\cong \{a,b,c\}$ and $1+1\cong \{0,1\}$, and  functions 
$\gamma_1, \gamma_2 \colon 1+1+1 \rightrightarrows 1+1$ defined as: 
\[
\gamma_1(a) = 0 \quad\quad \gamma_1(b) = \gamma_1 (c) = 1 \qquad\qquad \gamma_2(a) = \gamma_2 (c) = 1 \quad\quad \gamma_2(b) = 0
\]
Further assume we  have functions $f, g\colon X \rightrightarrows 1+1+1$  such that: \[\gamma_1 \circ f = \gamma_1 \circ g\qquad \qquad \qquad\gamma_2 \circ f = \gamma_2 \circ g\] We need to show that $f=g$. Suppose that $f\neq g$. Then $f(x)\neq g(x)$ for some $x\in X$. Assuming the existence of such $x$, we arrive to the following contradictions: 
\begin{itemize}
\item $f(x)= a  \Rightarrow g(x)\in \{b, c\} \Rightarrow  \gamma_1 (f(x))\neq \gamma_1 (g(x))$
\item $f(x)=b \Rightarrow  g(x)\in \{a, c\} \Rightarrow \gamma_2 (f(x))\neq \gamma_2 (g(x))$
\item $f(x)=c \Rightarrow g(x)\in \{a, b\}\Rightarrow  \gamma_1 (f(x))\neq \gamma_1 (g(x))$ if $g(x)=a$, or $\gamma_2 (f(x))\neq \gamma_2 (g(x))$ if $g(x)=b$
\end{itemize}
Hence it must be the case that $f=g$, and so $\gamma_1,~\gamma_2$ are jointly monic. 
\end{proof}

There are many examples of categories that are effectuses (for more, see \cite{cho2015introduction}):
\begin{itemize}
\item Topological spaces with continuous maps between them. 
\item Rings (with multiplicative identity) and ring homomorphisms.
\item Measurable spaces with measurable functions.
\item $C^{\ast}$-algebras with completely positive unital maps.
\item Extensive categories: $\B$ is extensive if it has finite coproducts and \[\B/X \times \B/Y \simeq \B/(X+Y)\] for all $X,Y$ objects of $\B$. (Every topos is extensive.)
\end{itemize}

 \section{Effectus of discrete probability measures}
For any set $X$ and any point $x\in X$, let $\mathbf{1}_x\colon X\to\{0,1\}$ denote the indicator function at $x$: 
\[\mathbf{1}_x(x'):=
\begin{cases}
1& x=x'\\
0& x\neq x
\end{cases}\qquad (x'\in X)
\]
A \emph{convex combination} of elements of the set $X$ is an expression: \[\lambda_1 \mathbf{1}_{x_1} + \cdots + \lambda_n\mathbf{1}_{x_n}\] with $n\geq 1$, $x_i\in X$, $\lambda_i\in [0,1]$, and  $\lambda_1+\cdots + \lambda_n = 1$. 
Let $D(X)$ be the set of all convex combinations of elements of $X$. 
 The elements of the set $D(X)$ are called  \emph{states} or \emph{distributions} on $X$.  The indicator function is a distribution $\mathbf{1}_x\in D(X)$, for every $x\in X$.
 A distribution $\sum_ip_i\mathbf{1}_{x_i}\in D(X)$ can be represented by a function $p\colon X \to [0,1]$ with finitely many non-zero values, satisfying: \[\sum_{x\in X}p(x)= 1\] Functions like this are called \emph{discrete probability measures}. If $\mathrm{range}(p)=\{x_1,\dots,x_n\},~n\geq 1$, then the assignment $p(x_i) \mapsto p_i$ gives the bijective correspondence between these two equivalent representations of distributions, \textit{i.e.}~as convex combinations or as discrete probability measures. 
   
Distributions on $X$ can be pushforwarded along a function $f\colon X \to Y$ to get distributions on $Y$: 
\[\sum_{x\in X} p_x\mathbf{1}_x \qquad\mapsto\qquad \sum_{x\in X} p_x\mathbf{1}_{f(x)}\]
 That is, we have a function $D (f)\colon D(X)\to D(Y)$ defined for any $p\in D(X)$ as:
  \[D (f)(p)(y):=
\sum_{x\in f^{-1}(y)}
p (x)\qquad (y\in Y)\]

 For every set $X$, the function $\eta_X\colon X \to D(X)$ is defined as: \[\eta_X (x):=\mathbf{1}_x\qquad (x\in X)\]   
 The function $\mu_X\colon D^2(X) \to D(X)$ is given by the expectation-value of evaluation functions $p\mapsto p(x)$ with respect to some distribution of distributions $P\in D^2(X)$, \textit{i.e.}~\[\mu_X(P)(x) := 
\sum_{p\in D(X)}P(p)\cdot p(x)\qquad(x\in X)\] 
The usual naturality and commutativity requirements are satisfied by these functions,  so there is a $[0,1]$-valued discrete distributions monad $D=(D, \eta,\mu)$ on $\Set$.

The Kleisli category $\mathcal{K}(D)$ of the distribution monad $D$ has sets as objects, and functions of type $X\to D(Y)$ as morphisms of type  $X\to Y$ in $\mathcal{K}(D)$. 
 The identity morphism $X\to X$ in $\mathcal{K}(D)$ is given by $\eta_X$. We define the Kleisli extension $c_{\ast}\colon D(X)\to D(Y)$ of any Kleisli morphism $c\colon X\to D(Y)$ as $c_{\ast}:=\mu_Y\circ D(c)$. That is,  for all $p\in D(X)$: \[c_{\ast}(p)(y)=\sum_{x\in X} p(x) \cdot c(x)(y)\qquad (y\in Y)\] Composition of  Kleisli maps   $c\colon X\to D(Y)$ and $d\colon Y\to D(Z)$, is given using Klesli extension (to simplify notation) as:
\[\xymatrix@C=90pt@R=20pt{
X \ar[r]^-{d_{\ast}\circ\ c\quad}_-{\quad=~\mu_Z\circ D(d)\circ c} & D(Z)
}\]
For any set function $f\colon X\to Y$, one can define a Kleisli map:
\[\xymatrix@C=90pt@R=20pt{
X \ar[r]^-{\hat{f}~:=~\eta_{Y}\circ f}_-{\quad=~D(f)\circ \eta_{X}}
& D(Y)
}\]
given by naturality of $\eta$.
 
The category $\mathcal{K}(D)$ has 
finite coproducts $(0, +)$ given by the empty set $0:=\emptyset$, and disjoint union $X_1+X_2$ with coprojections: \[ \xymatrix@C=90pt@R=20pt{
X_i \ar[r]^-{\hat{\kappa_i}~:=~\eta_{(X_1+X_2)}\circ\kappa_i}_-{\quad=~D(\kappa_i)\circ \eta_{X_i}}
& D(X_1+X_2)
}\qquad (i=1,2)\] 
The following result is well-known, see \textit{e.g.}~\cite[Proposition 6.4]{jacobs2011probabilities}. 

\begin{lemma} \label{lem:aff}
The distribution monad $D$ is affine, \textit{i.e.}~$D(1)\cong 1$. Moreover, $D(1+1)\cong [0,1]$.
\end{lemma}
\begin{proof}
An element $p\in D(1)$ can be regarded as a function $p\colon 1 \to [0,1]$ such that $\sum_{x\in 1}p(x)=1$. Therefore, it must be the case that $p$ is the constant function $1$. Thus $D(1)=\{1\}\cong 1$. Now, a $[0,1]$-valued distribution over a $2$-element set consists of a choice of $p\in [0,1]$ for one element and $1-p$ for the other element. Hence $D(1+1)\cong [0,1]$.
\end{proof}

\begin{proposition}
 $\mathcal{K}(D)$ has a terminal object.
 \end{proposition}
 \begin{proof}
 By Lemma~\ref{lem:aff}, any choice of a singleton set $1$ in $\Set$  is a terminal object in $\mathcal{K}(D)$. We have unique arrows: 
\[\xymatrix@C=90pt@R=20pt{
X \ar[r]^-{\hat{!}_X~:=~\eta_{1}\circ!_X}_-{\qquad=~ D(!_X)\circ\eta_X}
& D(1)
}\]
for any $X$ in $\Set$. 
Since $1\cong D(1)$, the unit  $\eta_1\colon 1 \to D(1)$ and the identity function $\id_1\colon 1 \to 1$ are equal $\eta_1=\id_1$. Therefore, we have $\hat{!}_X =!_X$ for all $X$ in $\Set$. 
 \end{proof}

\begin{proposition}
$\mathcal{K}(D)$ has pullbacks.
\end{proposition}
\begin{proof}
Assume we have the following commutative diagram: 
\[ \xymatrix@C=50pt@R=20pt{
A\ar@/^1pc/[drr]^{d}\ar@/_2pc/[ddr]_{c} \ar@{.>}[dr]^-{u}&&\\
&P\ar[r]^-{i}\ar[d]_-{h}&Y\ar[d]^{g}\\
&X\ar[r]_-{f}&Z
}\]
 where all the arrows live in $\mathcal{K}(D)$, and the dashed arrow means that $u$ is uniquely defined. That is, we have a function $u\colon A \rightarrow D(P)$ which is determined in a unique way by  Kleisli maps:
 \begin{eqnarray*}
  c\colon A \rightarrow D(X)\quad& f\colon X \rightarrow D(Z)\quad& h\colon P \rightarrow D(X)\\
  d\colon A \rightarrow D(Y)\quad& g\colon Y\rightarrow D(Z)\quad& i\colon P \rightarrow D(Y)
 \end{eqnarray*}
 satisfying the following four equations:
 \begin{eqnarray}
h_{\ast}\circ u&=&c\label{p:1}\\
 i_{\ast}\circ u&=&d\label{p:2}\\
 f_{\ast}\circ c&=&g_{\ast}\circ d\label{p:3}\\
 f_{\ast}\circ h&=&g_{\ast}\circ i\label{p:4}
 \end{eqnarray}
 In that case, we have that $P$ is the pullback of $g$ along $f$ in $\mathcal{K}(D)$. 
\end{proof}

The following and last result is well-known, see \emph{e.g.}~\cite[Example 4.7]{jacobs2015new}.

\begin{theorem}\label{prop:distr}
The Kleisli category $\mathcal{K}(D)$ of the distribution monad $D$ on sets is an effectus.
\end{theorem} 
\begin{proof}
We need to check two pullback conditions and one joint monicity requirement for the Kleisli category $\mathcal{K}(D)$. We start with the first pullback from Definition~\ref{def:eff}. We assume to have the following Kleisli maps:
 \begin{eqnarray*}
  c\colon A \rightarrow D(X+1)\quad& f\colon X+1 \rightarrow D(1+1)\quad& h\colon X+Y \rightarrow D(X+1)\\
  d\colon A \rightarrow D(1+Y)\quad& g\colon 1+Y \rightarrow D(1+1)\quad& i\colon X+Y \rightarrow D(1+Y)
 \end{eqnarray*}
 where: 
 \begin{eqnarray*}
 f:=D(!_X + \id_1) \circ \eta_{X+1}&\qquad&h:=D(\id_X + !_Y) \circ \eta_{X+Y}\\
 g:=D(\id_1 + !_Y) \circ \eta_{1+Y}&\qquad&i:=D(!_X + \id_Y) \circ \eta_{X+Y}
 \end{eqnarray*}
By definition of Kleisli extension we have:
\[\begin{split}
f_{\ast} &= \mu_{1+1} \circ D(f)\\
&=\mu_{1+1} \circ D(D(!_X+\id_1)\circ\eta_{X+1})\\
&\overset{\star}{=} \mu_{1+1}\circ D(\eta_{1+1} \circ (!_X + \id_1))\\
&=\mu_{1+1}\circ D(\eta_{1+1}) \circ D(!_X + \id_1)\\
&=D(!_X + \id_1)
\end{split}\]
where the marked equality $\overset{\star}{=}$ follows from naturality of $\eta$, and the last one from the axioms of monads. Similarly, we have: 
\[\begin{split}
g_{\ast} &= D(\id_1 + !_Y)\\
h_{\ast} &= D(\id_X + !_Y)\\
i_{\ast} &= D(!_X + \id_Y)
\end{split}\] 
Therefore, equation~\eqref{p:4} above holds:
 \[\begin{split}
 f_{\ast}\circ h&=D(!_X + \id_1) \circ h\\
 &=D(!_X + \id_1) \circ D(\id_X + !_Y) \circ \eta_{X+Y}\\
  &=D((!_X + \id_1) \circ(\id_X + !_Y)) \circ \eta_{X+Y}\\
    &\overset{\star}{=}D((\id_1 + !_Y) \circ(!_X + \id_Y)) \circ \eta_{X+Y}\\
     &=D(\id_1 + !_Y) \circ D(!_X + \id_Y) \circ \eta_{X+Y}\\
 &= g_{\ast}\circ i
 \end{split}\]
 where the marked equality $\overset{\star}{=}$ follows from the fact that both squares in the definition of effectus (see Definition~\ref{def:eff}) commute in every category with finite coproducts and a terminal object. 
  
  Let $X+1=X+\{1\}$ and $1+Y =\{0\} + Y$. Further suppose  the Kleisli maps $c\colon A \to D(X+\{1\})$ and $d\colon A \to D(\{0\}+Y)$ satisfy equation~\eqref{p:3} above. More concretely, suppose: 
\begin{equation} \label{eq:proof2.1}
D(!_X+\id_Y)(c (a))= D(\id_X + !_Y) ( d (a))\quad \in \quad D(\{0\}+\{1\})
\end{equation}  
for all $a\in A$. 
Specifically, this equation~\eqref{eq:proof2.1} expanded and evaluated says that: 
\begin{equation}\begin{split}
D(!_X+\id_Y)(c(a))(0) 
&\overset{\eqref{eq:proof2.1}}{=} 
D(\id_X + !_Y)(d(a))(0)\\
&=\sum_{y\in (\id_X+!_Y)^{-1}(0)}d(a)(y) \\
&=
d(a)(0)
 \quad \in \quad [0,1] \label{eq:da02}
 \end{split}\end{equation}
 \begin{equation}\begin{split}
D(!_X+\id_Y)(c(a))(1) &= \sum_{x\in (!_X+ \id_Y)^{-1}(1)} c(a)(x)  \\
&=  c(a)(1) %\\
 \quad \in \quad [0,1] \label{eq:ca12}
\end{split}\end{equation}
Thus:
\begin{equation}\begin{split}
 d(a)(0) + c(a)(1) = 1 \label{eq:thus2}
\end{split}\end{equation} 
We have: 
\begin{equation*}\begin{split}
\sum_{x\in X} u(a)(x) + \sum_{y\in Y} u(a)(y)
&\overset{\textrm{def}}{=}
\sum_{x\in X} c(a)(x) + \sum_{y\in Y} d(a)(y)   \\
&=
\sum_{!_X(x)=0} c(a)(x)  + \sum_{!_Y(y)=1} d(a)(y) \\
&=
D(!_X+\id_Y)(c(a))(0)  + D(!_X+\id_Y)(c(a))(1)\\
&\overset{\star}{=}
d(a)(0) + c(a)(1)\\
&{=}
1
\end{split}\end{equation*}
where the marked equality $\overset{\star}{=}$ follows from~\eqref{eq:da02} and~\eqref{eq:ca12}, and the last equality from~\eqref{eq:thus2}.
Hence $u$ is well-defined. We still need to check~\eqref{p:1} and~\eqref{p:2} above, which in this case amounts to show that:
	\[\begin{split} 
	D(\id_X + !_Y) \circ u &= c  \label{eq:pbk12}\\
	D(!_X + \id_Y) \circ u &= d 
	\end{split}\] 
 For all $a\in A$, we have indeed:
\begin{equation*}\begin{split}
	D(\id_X + !_Y) (u(a))(x) &= 
	\sum_{x'\in(\id_X+!_Y)^{-1}(x)}
	 u(a)(x')\\
	&=
	u(a)(x)\\
	&\overset{\textrm{def}}{=}
	c(a)(x)\\
	D(\id_X + !_Y) (u(a))(1) &= \sum_{y\in(\id_X+!_Y)^{-1}(1)} u(a)(y)\\
	&=
	\sum_{y\in Y} u(a)(y)\\
	&\overset{\textrm{def}}{=}
	\sum_{y\in Y} d(a)(y)\\
	&=
	c(a)(1)
\end{split}
\end{equation*}
\begin{equation*}\begin{split}	
D(!_X + \id_Y) (u(a))(y) &= 
	\sum_{y'\in(!_X+\id_Y)^{-1}(y)} u(a)(y')\\
	&=
	u(a)(y)\\
	&\overset{\textrm{def}}{=}
	d(a)(y)\\
	D(!_X + \id_Y) (u(a))(0) &= \sum_{x\in(!_X+\id_Y)^{-1}(0)} u(a)(x)\\
	&=
	\sum_{x\in X} u(a)(x)\\
	&\overset{\textrm{def}}{=}
	\sum_{x\in X}c(a)(x)\\
	&=
	d(a)(0).
	\end{split}\end{equation*}
By definition, $u\colon A \to D(X + Y)$ is the unique Kleisli map satisfying the needed requirements. This completes the proof of the first pullback condition for $\mathcal{K}(D)$.

For the second pullback from Definition~\ref{def:eff}, 
let $1:=\{0\}$ and $1:=\{1\}$ be two distinct (choices of) singleton sets, and let $D(1)\cong\{\mathbf{1}_0\}$. % $ 
Consider Kleisli maps $!_A\colon A \to \{\mathbf{1}_0\}$ and $c\colon A \to D(X+Y)$ such that:
\begin{equation}\label{eq:jim2}
D(!_X+!_Y) \circ c   = D ( \kappa_1) \circ\ !_A%
\end{equation}
Since $c(a)=\sum_{x} p_x \mathbf{1}_x  + \sum_{y} p_y \mathbf{1}_y \in D(X+Y)$ with  $\sum_{x} p_x  + \sum_{y} p_y =1 \in [0,1]$ for all $a\in A$, we have that the left-hand side of equation~\eqref{eq:jim2} expands to:
\begin{equation*}\label{eq:this}
D(!_X+!_Y) (c(a)) = \sum_{x} p_x \mathbf{1}_0 + \sum_{y} p_y \mathbf{1}_1 
\end{equation*}
The right-hand side of equation~\eqref{eq:jim2}expands to:
\begin{equation*}\begin{split}
D( \kappa_1)(!_A(a)) 
&\overset{}{=}
 D ( \kappa_1)(\mathbf{1}_0) \\
 &= 
 \mathbf{1}_{\kappa_1(0)}\\
 &=
  \mathbf{1}_0
  \label{eq:that}
  \end{split}\end{equation*}
Hence $\sum_{x} p_x = 1$, and so $c(a)\in D(X)$. Let $u\colon A \to D(X)$ be defined as $u(a)(x):= c(a)(x)$. By definition, the Kleisli map $u\colon A \to D(X)$ is the unique arrow satisfying the needed requirements.

Now we prove that the maps $\gamma_1, \gamma_2 \colon (1+1)+1 \rightrightarrows 1 +1$ in are jointly monic in $\mathcal{K}(D)$. This part is taken exactly from~\cite[Example 4.7]{jacobs2015new}. Let $\sigma,\tau\in D(1+1+1)$ be distributions such that 
\begin{equation}\begin{split}
D(\gamma_1)(\sigma) &= D(\gamma_1)(\tau) \label{eq:jo12}\\
D(\gamma_2)(\sigma) &= D(\gamma_2)(\tau)% 
\end{split}\end{equation}
 in $D(1+1)$. Assume $1+1+1=\{a,b,c\}$ and $1+1=\{0,1\}$. We have the following convex combinations for $\sigma$ in $D(1+1)$:
\begin{equation*}\begin{split} 
D(\gamma_1)(\sigma)
&=
 \sigma(a) \mathbf{1}_0 + (\sigma(b)+\sigma(c))\mathbf{1}_1\\
 D(\gamma_2)(\sigma) 
 &=
 \sigma(b)\mathbf{1}_0 + (\sigma(a)+\sigma(b))\mathbf{1}_1
\end{split}\end{equation*}
Similarly for $\tau$:
\begin{equation*}\begin{split} 
D(\gamma_1)(\tau)
&=
 \tau(a)  \mathbf{1}_0 + (\tau(b)+\tau(c)) \mathbf{1}_1\\
 D(\gamma_2)(\tau) 
 &=
 \tau(b) \mathbf{1}_0 + (\tau(a)+\tau(b)) \mathbf{1}_1
\end{split}\end{equation*}
Hence, by the first equation in~\eqref{eq:jo12}, we have $\sigma(a)=\tau(a)$. Similarly, by the second equation in~\eqref{eq:jo12}, we have $\sigma(b)=\tau(b)$. We still need to show that $\sigma(c)=\tau(c)$. Since $\sigma(a)+\sigma(b)+\sigma(c)=1=\tau(a)+\tau(b)+\tau(c)$, then:
\[\begin{split}
\sigma(c) &= 1 - 
(\sigma(a) + \sigma(b))\\
&=  1 -
(\tau(a) + \tau(b))\\
&=
\tau(c)
\end{split}\] 
\end{proof}

\section{Effectus of projection-valued measures}

A simple undirected graph $G$ is a relational structure with a single, binary irreflexive and symmetric relation $E(G)$ that we have been written as $\sim$ in infix notation.  
Relational structures are more general. 
 A \emph{relational structure} $\mathscr{A}=(A,R(\mathscr{A}))$ consists of a set $A$ together with an indexed family $R(\mathscr{A})=(R^{\mathscr{A}}_{i})_{i\in I}$ of relations $R^{\mathscr{A}}_{i}\subseteq A^{k_i}$ with $I$ a set of indices, and $k_i\geq 1$ for all $i\in I$. 
 A \emph{homomorphism} of relational structures $\mathscr{A} \to \mathscr{B}$ is a function $f\colon A \to B$ between the underlying sets, preserving all relations: $(x_1,\dots,x_k)\in R^{\mathscr{A}} \Rightarrow (f(x_1),\dots,f(x_k))\in R^{\mathscr{B}} $ for all $(x_1,\dots,x_k)\in A^k$ and all $R^{\mathscr{A}}\in R(\mathscr{A})$ with arity $k\geq 1$.  There is a category  $\RStr$ of relational structures and homomorphisms between them.

 We write $\mathrm{Proj}(d)\subseteq M_d(\mathbb{C})$ for the set of all $d\times d$ complex matrices that are both self-adjoint and idempotent, \textit{i.e.}~$\mathrm{Proj}(d)=\{a\in M_d(\mathbb{C}) : a^{\ast}=a^2=a\}$ for all $d\geq 1$.  The elements of $\mathrm{Proj}(d)$ are called \emph{$d$-dimensional projections}.  
 For every $d\geq 1$, the identity matrix is a projection $\mathbbm{1}\in\mathrm{Proj}(d)$.
 
 For every relational structure $\mathscr{A} =(A, R(\mathscr{A}))$ and every positive integer $d\geq 1$,  define a relational structure $Q_d
(\As)= (Q_d(A), R({Q_d(\mathscr{A})})$, where
 \[Q_d(A):=\{\sum_{x\in A}p_x\mathbf{1}_x : p_x\in\mathrm{Proj}(d), \sum_{x\in A} p_x=\mathbbm{1}\} \]
and every $k$-ary relation $R^{Q_d(\mathscr{A})}\subseteq Q_d(A)^k$ in $R(Q_d(\mathscr{A}))$ is defined as the set of $k$-tuples 
$(p_1,\dots,p_k)$ 
of projection-valued distributions $p_1,\dots, p_k \in Q_d(A)$ satisfying:
\begin{itemize}
 \item[(1)] $p_i(x)$ and $p_j(x')$ commute for all $x,x'\in A$
 \item[(2)] $(x_1,\dots,x_k)\notin R^{\mathscr{A}}$ implies $\prod_{i=1}^k p_i(x_i)=0$ for all $(x_1,\dots,x_k)\in A^k$
 \end{itemize} 
There cannot be infinitely many projections resolving the $d$-dimensional identity. Therefore, every projection-valued distribution $p\in Q_d(A),~p\colon A\to\mathrm{Proj}(d)$, has finitely manny non-zero values.  These distributions are  \emph{projection-valued measures} (PVMs) from functional analysis and quantum theory~\cite{heinosaari2011mathematical, chang_2015}.

For any homomorphism $f\colon \mathscr{A} \rightarrow \mathscr{B}$, we define the homomorphism $Q_d(f)\colon Q_d(\mathscr{A}) \rightarrow Q_d(\mathscr{B})$  as 
 \[Q_d(f)(p)(y):=\sum_{x\in f^{-1}(y)}p(x)\qquad (y\in B)\]  or, equivalently, as $Q_d(f)(p):=\sum_x p_x\mathbf{1}_{f(x)}\in Q_d(\mathscr{B})$.  
 This definition preserves composites and identities, so there is a functor $Q_d\colon \RStr \rightarrow \RStr$ for every $d\in\mathbb{N}=\{1,2,\dots\}$ (see~\cite{abramsky2017quantum} for more details). 

Note that $\mathrm{Proj}(1)\cong \{0,1\}\cong 1 + 1 $. We define $\eta_{\As} \colon \mathscr{A} \rightarrow Q_1(\mathscr{A})$ to be the indicator function $\mathbf{1}_x$ at $x\in A$: $\eta_{\As}(x)(x')=1$ if $x=x'$ and $\eta_{\As}(x)(x')=0$ if $x\neq x'$.  Next we use the tensor product of matrices. 
For all $d,d'\geq 1$, let $\mu^{d,d'}_{\As}\colon Q_d Q_{d'}(\mathscr{A}) \rightarrow Q_{dd'}(\mathscr{A})$ be defined for any $P\colon Q_{d'}(A)\to\mathrm{Proj}(d)$ as: \[\mu^{d,d'}_{\As}(P)(x):=\sum_{p\in Q_{d'}(A)} P(p) \otimes p(x)\qquad (x\in A)\] 
 These maps $\eta_{\As},~\mu^{d,d'}_{\As}$ are   components of  natural transformations $\eta \colon 1 \Rightarrow Q_1,~\mu^{d,d'}\colon Q_d Q_{d'} \Rightarrow Q_{dd'}$ satisfying the axioms of graded monads~\cite{milius2015generic}:
\[\mu^{d,1}_{\As} \circ Q_d (\eta_{\As})=\id_{Q_d(\As)}= \mu^{1,d}_{\As}\circ \eta_{Q_d(\As)} \quad\text{and}\quad\mu^{d,d'd''}_{\As}\circ Q_d(\mu^{d', d''}_{\As})= \mu^{dd',d''}_{\As} \circ \mu^{d,d'}_{Q_{d''}(\As)}\]

 Given $\mathscr{A},\mathscr{B}$ relational structures, $\mathscr{A}+\mathscr{B}$ is the relational structure over the set $A+B$ where the $k$-ary relation $R^{\mathscr{A}+\mathscr{B}}$ is defined as all the tuples $(x_1, \dots,x_k)\in (A+B)^k$ satisfying either $(x_1, \dots,x_k)\in R^{\mathscr{A}}$ or $(x_1, \dots,x_k)\in R^{\mathscr{B}}$. Also, we have the structure $0$ over the empty set $\emptyset =0$ with no relations. Further we have a structure $1$ over some singleton set $1=\{\ast\}$ with the universal relation of arity $k$, \textit{i.e.} one has $R^1:=1^k=1\times\cdots\times 1$.  Like the distribution monad $D$, the quantum monad $Q_d$ is also an affine monad since $Q_d(1)\cong 1$. The following result is immediate: 
 \begin{proposition}
$\mathcal{K}(Q_d)$ has a terminal object and finite coproducts. 
\end{proposition}

Recall that semirings are rings without all additive inverses. That is, a \emph{semiring} consists of a set $S$ and two binary operations $+,\cdot\colon S\times S \to S$ called addition and multiplication, such that $(S,+)$ is a commutative monoid, $(S,\cdot)$ is a monoid, and there is a distributive law of multiplication over addition. A \emph{partial semiring} is a semiring where at least one binary operation is partially defined, \emph{i.e.}~the binary operation is a partial function. The real unit interval $[0,1]$ is a partial semiring with addition $x+y$ defined only when $x+y\leq 1$. For any $d\geq 1$, the set of $d$-dimensional projections $\mathrm{Proj}(d)$ is a partial semiring with addition $p+q$ defined only when $p\cdot q = 0$.

\begin{theorem} \label{thm:effKlQ}
The Kleisli category $\mathcal{K}(Q_d)$ of the quantum monad $Q_d$ is an effectus.
\end{theorem}
\begin{proof}
The proof of Theorem~\ref{prop:distr} works for any $S$-valued distributions monad with $S$ partial semiring  since we did not use any fact about the unit interval $[0,1]$ that do not hold for any other partial semiring. There we saw the details for $S=[0,1]$, and here we leave those for $S=\mathrm{Proj(d)}$ to the reader.  Similarly, the two pullback conditions and one joint monicity requirement from Definition~\ref{def:eff} hold for the Kleisli category $\mathcal{K}(Q_d)$ forgetting the homomorphism part (\emph{i.e.}~preserving relations).  Thus, all that carries the same at the level of sets. Now we mention the parts about homomorphism. Let $u\colon \mathscr{P} \to Q_d(\As + \Bs)$ be the Kleisli map defined (in the first pullback condition) for each $a\in P$ as $u(a)(x):= c(a)(x) \in \mathrm{Proj}(d)$ for all $x\in A$, and $u(a)(y):= d(a)(y) \in \mathrm{Proj}(d)$ for all $y\in B$, where $c\colon \mathcal{P}\rightarrow Q_d(\As + 1)$ and $d\colon  \mathcal{P} \rightarrow Q_d(1+\Bs)$ are given homomorphisms.  This Kleisli map $u$ is  homomorphism by definition, since both $c$ and $d$ are homomorphisms by assumption. For the second pullback condition, now we suppose to have  a homomorphism $c\colon \mathcal{P}\rightarrow Q_d(\As+\Bs)$ and define $u\colon \mathcal{P}\rightarrow Q_d(\As)$ as $u(a)(x):= c(a)(x)$ for all $a\in P$ and $x\in A$. Once again, here we have that $u$ is homomorphism by definition since $c$ is homomorphism by assumption. 
\end{proof}

\section{Quantum probabilistic reasoning}

States in $\mathcal{K}(D)$ are discrete probability measures. In  $\mathcal{K}(Q_d)$, states are quantum measurments (PVMs) also known as \emph{sharp observables} \cite{heinosaari2011mathematical}.  
We shall start describing what is the situation with respect to states and predicates in general for an arbitrary effectus $\B$. Formally, a \emph{state} on $X$ is a morphism in $\B$ with type $1\to X$. A \emph{predicate} on $X$ is a morphism in $\B$ with type $X\to 1+1$. 
There is an adjunction:
\[\xymatrix@C=30pt@R=20pt{
\Pred(\B)^{\mathrm{op}}\ar@/^8pt/[rr]^{\Stat} &\top&\Stat(\B)\ar@/^8pt/[ll]^{\Pred} \\
&\B\ar[ur]\ar[ul]&
}\]
where $\B\rightarrow\Stat(\B)$ and $\B\rightarrow\Pred(\B)^{\mathrm{op}}$ are the functors defined on objects as $\Stat(X):=\B(1, X)$ and $\Pred(X):=\B(X, 1+1)$, for any $X$ object of $\B$; the action of these functors on a given morphism $f\colon X \rightarrow Y$ in $\B$ produce morphisms $\Stat(f)\colon \Stat(X) \rightarrow \Stat(Y)$  and $\Pred(f)\colon \Pred(Y) \rightarrow \Pred(X)$ called \emph{state} and \emph{predicate transformer} 
 defined by compositon in $\B$
as 
 $\Stat(f)(p):=f\circ p$ and $\Pred(f)(q):=q\circ f$, for any $p\in\Stat(X)$ and $q\in\Pred(Y)$. 
 
 Morphisms in $\Pred(1)=\Stat(1+1) =\B(1, 1+1)$ are called \emph{scalars}.  For probability theory, scalars are probabilities $\mathcal{K}(D)(1, 1+1)\cong [0,1]$. 
Given a state $p\in\Stat(X)$ and a predicate $q\in \Pred(X)$ on the same object $X$ of $\B$, we have by definition $\Stat(q)(p)=\Pred(p)(q)=q\circ p$. 
%The \emph{validity} $p\models q$ of $q$ in $p$ is defined to be the morphism $q\circ p\in\B(1, 1+1)$.   
At the level of sets, scalars for the quantum case are projections $\mathcal{K}(Q_d)(1, 1+1)\cong  \mathrm{Proj}(d)$.  Thus, predicates are assignments of projections. Now, let's consider $\mathrm{Proj}(d)$ as relational structure of projections with a $k$-ary relation $R^{\mathrm{Proj}(d)}$ given by $(p_1,\dots,p_k)\in R^{\mathrm{Proj}(d)} $ if and only if the projections $p_1,\dots,p_k$ pairwise commute: $p_i\cdot p_j = p_j \cdot p_i$ for all $i,j=1,\dots,k$.  

\begin{proposition}
Let $\As=(A, R^{\As})$ be a relational structure. Then:
\begin{itemize}
\item a \emph{state} on $\As$ is a PVM $p\in Q_d(\As)$ on the underlying set $A$;  
\item a \emph{predicate} $q\colon \As \to \mathrm{Proj}(d)$ is an assignment of projections $q\colon A \to \mathrm{Proj}(d)$ such that points appearing in some tuple in the relation $R^\As$ get assigned commuting projections, \textit{i.e.}~projections $q(x_i), q(x_j)\in \mathrm{Proj}(d)$ commute if there exists $\mathbf{x}\in A^k$ such that: \[\mathbf{x}=(x_1,\dots, x_i,\dots,x_j,\dots,x_k),\quad \mathbf{x}\in R^{\As}\]
\end{itemize}
\end{proposition}

Given a PVM and an assignment of commuting projections, we can compute their validity, or expected value, of the projections in the measure. 

\begin{proposition}
Let $p=(p_x:x\in A),~p_x\in\mathrm{Proj}(d),$ be a PVM and $q=(q_x:x\in A),~q_x\in\mathrm{Proj}(d'),$ a collection of pairwise commuting projections. 
\begin{itemize}
\item \emph{Validity of $q$ in $p$} is the projection $p\models q\in\mathrm{Proj}(dd')$ defined as: 
\[\begin{split}
p\models q&:=\sum_{x\in A} p_x \otimes q_x
\end{split}\]
\item \emph{Conditioning $p$ given $q$} is the PVM $p|_q\in Q_{dd'}(\As)$ defined if validity $p \models q$ is non-zero as:
\[
%p|_q:=\sum_{x\in A}\frac{p_x\otimes q_x}{p\models q} \mathbf{1}_x
p|_q=\left(\frac{p_x\otimes q_x}{p\models q}:x\in A\right)
\]  
\end{itemize}
\end{proposition}

\noindent The rows of the following table correspond to scalars, states, predicates, and validity in three different effectuses (deterministic, probabilistic, quantum):
 
 \begin{table}[htp]
\begin{center}
\begin{tabular}{|c|c|c|c|}\hline
$\mathbf{Effectus}$&$\mathbf{Set}$&$\mathcal{K}(D)$&$\mathcal{K}(Q_d)$\\\hline
$1\to 1+1$&$b\in\{0,1\}$&$p\in[0,1]$&$E\in\mathrm{Proj}(d)$\\\hline
$1\to X$&$x\in X$&$(p_x)_{x\in X},~\sum p_x = 1$&$(E_x)_{x\in X},~\sum E_x = \mathbbm{1}$\\\hline
$X\to 1+1$&$S\subseteq X$&$f\colon X\to [0,1]$&$q\colon X\to \mathrm{Proj}(d)$\\\hline
$1\to X\to 1+1$&$x\in S$&$\sum p_x\cdot f(x)$&$\sum E_x\otimes q(x)$\\\hline
\end{tabular}
\end{center}
\label{default}
\end{table}%+
 
 \section*{Acknowledgements}
We thank Samson Abramsky, Rui Soares Barbosa,  Aleks Kissinger, Sandra Palau, Matteo Sammartino, and Fabio Zanasi, as well as anonymous referees for all the wise comments.

\bibliographystyle{eptcs}
\bibliography{paper63}

\begin{thebibliography}{10}
\providecommand{\bibitemdeclare}[2]{}
\providecommand{\surnamestart}{}
\providecommand{\surnameend}{}
\providecommand{\urlprefix}{Available at }
\providecommand{\url}[1]{\texttt{#1}}
\providecommand{\href}[2]{\texttt{#2}}
\providecommand{\urlalt}[2]{\href{#1}{#2}}
\providecommand{\doi}[1]{doi:\urlalt{http://dx.doi.org/#1}{#1}}
\providecommand{\bibinfo}[2]{#2}

\bibitemdeclare{inproceedings}{abramsky2017quantum}
\bibitem{abramsky2017quantum}
\bibinfo{author}{Samson \surnamestart Abramsky\surnameend},
  \bibinfo{author}{Rui~Soares \surnamestart Barbosa\surnameend},
  \bibinfo{author}{Nadish \surnamestart de~Silva\surnameend} \&
  \bibinfo{author}{Octavio \surnamestart Zapata\surnameend}
  (\bibinfo{year}{2017}): \emph{\bibinfo{title}{The Quantum Monad on Relational
  Structures}}.
\newblock In: {\sl \bibinfo{booktitle}{LIPIcs-Leibniz International Proceedings
  in Informatics}}, \bibinfo{volume}{83}, \bibinfo{organization}{Schloss
  Dagstuhl-Leibniz-Zentrum fuer Informatik}, \doi{10.4230/LIPIcs.MFCS.2017.35}.

\bibitemdeclare{article}{atserias2018quantum}
\bibitem{atserias2018quantum}
\bibinfo{author}{Albert \surnamestart Atserias\surnameend},
  \bibinfo{author}{Laura \surnamestart Man{\v{c}}inska\surnameend},
  \bibinfo{author}{David~E \surnamestart Roberson\surnameend},
  \bibinfo{author}{Robert \surnamestart {\v{S}}{\'a}mal\surnameend},
  \bibinfo{author}{Simone \surnamestart Severini\surnameend} \&
  \bibinfo{author}{Antonios \surnamestart Varvitsiotis\surnameend}
  (\bibinfo{year}{2018}): \emph{\bibinfo{title}{Quantum and non-signalling
  graph isomorphisms}}.
\newblock {\sl \bibinfo{journal}{Journal of Combinatorial Theory, Series B}},
  \doi{10.1016/j.jctb.2018.11.002}.

\bibitemdeclare{article}{cameron2007quantum}
\bibitem{cameron2007quantum}
\bibinfo{author}{Peter~J \surnamestart Cameron\surnameend},
  \bibinfo{author}{Ashley \surnamestart Montanaro\surnameend},
  \bibinfo{author}{Michael~W \surnamestart Newman\surnameend},
  \bibinfo{author}{Simone \surnamestart Severini\surnameend} \&
  \bibinfo{author}{Andreas \surnamestart Winter\surnameend}
  (\bibinfo{year}{2007}): \emph{\bibinfo{title}{On the quantum chromatic number
  of a graph}}.
\newblock {\sl \bibinfo{journal}{The electronic journal of combinatorics}}
  \bibinfo{volume}{14}(\bibinfo{number}{1}), p.~\bibinfo{pages}{81}.
\newblock \urlprefix\url{https://arxiv.org/abs/quant-ph/0608016}.

\bibitemdeclare{book}{chang_2015}
\bibitem{chang_2015}
\bibinfo{author}{Mou-Hsiung \surnamestart Chang\surnameend}
  (\bibinfo{year}{2015}): \emph{\bibinfo{title}{Quantum Stochastics}}.
\newblock \bibinfo{series}{Cambridge Series in Statistical and Probabilistic
  Mathematics}, \bibinfo{publisher}{Cambridge University Press},
  \doi{10.1017/CBO9781107706545}.

\bibitemdeclare{article}{cho2015introduction}
\bibitem{cho2015introduction}
\bibinfo{author}{Kenta \surnamestart Cho\surnameend}, \bibinfo{author}{Bart
  \surnamestart Jacobs\surnameend}, \bibinfo{author}{Bas \surnamestart
  Westerbaan\surnameend} \& \bibinfo{author}{Abraham \surnamestart
  Westerbaan\surnameend} (\bibinfo{year}{2015}): \emph{\bibinfo{title}{An
  introduction to effectus theory}}.
\newblock {\sl \bibinfo{journal}{arXiv preprint}}.
\newblock \urlprefix\url{https://arxiv.org/abs/1512.05813}.

\bibitemdeclare{inproceedings}{cleve2014characterization}
\bibitem{cleve2014characterization}
\bibinfo{author}{Richard \surnamestart Cleve\surnameend} \&
  \bibinfo{author}{Rajat \surnamestart Mittal\surnameend}
  (\bibinfo{year}{2014}): \emph{\bibinfo{title}{Characterization of binary
  constraint system games}}.
\newblock In: {\sl \bibinfo{booktitle}{International Colloquium on Automata,
  Languages, and Programming}}, \bibinfo{organization}{Springer}, pp.
  \bibinfo{pages}{320--331}, \doi{10.1007/978-3-662-43948-7}.

\bibitemdeclare{incollection}{giry1982categorical}
\bibitem{giry1982categorical}
\bibinfo{author}{Michele \surnamestart Giry\surnameend} (\bibinfo{year}{1982}):
  \emph{\bibinfo{title}{A categorical approach to probability theory}}.
\newblock In: {\sl \bibinfo{booktitle}{Categorical aspects of topology and
  analysis}}, \bibinfo{publisher}{Springer}, pp. \bibinfo{pages}{68--85},
  \doi{10.1007/BFb0092872}.

\bibitemdeclare{book}{heinosaari2011mathematical}
\bibitem{heinosaari2011mathematical}
\bibinfo{author}{Teiko \surnamestart Heinosaari\surnameend} \&
  \bibinfo{author}{M{\'a}rio \surnamestart Ziman\surnameend}
  (\bibinfo{year}{2011}): \emph{\bibinfo{title}{The mathematical language of
  quantum theory: from uncertainty to entanglement}}.
\newblock \bibinfo{publisher}{Cambridge University Press},
  \doi{10.1017/CBO9781139031103}.

\bibitemdeclare{article}{jacobs2011probabilities}
\bibitem{jacobs2011probabilities}
\bibinfo{author}{Bart \surnamestart Jacobs\surnameend} (\bibinfo{year}{2011}):
  \emph{\bibinfo{title}{Probabilities, distribution monads, and convex
  categories}}.
\newblock {\sl \bibinfo{journal}{Theoretical Computer Science}}
  \bibinfo{volume}{412}(\bibinfo{number}{28}), pp. \bibinfo{pages}{3323--3336},
  \doi{10.1016/j.tcs.2011.04.005}.

\bibitemdeclare{article}{jacobs2015new}
\bibitem{jacobs2015new}
\bibinfo{author}{Bart \surnamestart Jacobs\surnameend} (\bibinfo{year}{2015}):
  \emph{\bibinfo{title}{New directions in categorical logic, for classical,
  probabilistic and quantum logic}}.
\newblock {\sl \bibinfo{journal}{Logical Methods in Computer Science (LMCS)}}
  \bibinfo{volume}{11}(\bibinfo{number}{3}), \doi{10.2168/LMCS-11(3:24)2015}.

\bibitemdeclare{article}{manvcinska2016quantum}
\bibitem{manvcinska2016quantum}
\bibinfo{author}{Laura \surnamestart Man{\v{c}}inska\surnameend} \&
  \bibinfo{author}{David~E \surnamestart Roberson\surnameend}
  (\bibinfo{year}{2016}): \emph{\bibinfo{title}{Quantum homomorphisms}}.
\newblock {\sl \bibinfo{journal}{Journal of Combinatorial Theory, Series B}}
  \bibinfo{volume}{118}, pp. \bibinfo{pages}{228--267},
  \doi{10.1016/j.jctb.2015.12.009}.

\bibitemdeclare{inproceedings}{milius2015generic}
\bibitem{milius2015generic}
\bibinfo{author}{Stefan \surnamestart Milius\surnameend}, \bibinfo{author}{Dirk
  \surnamestart Pattinson\surnameend} \& \bibinfo{author}{Lutz \surnamestart
  Schr{\"o}der\surnameend} (\bibinfo{year}{2015}):
  \emph{\bibinfo{title}{Generic trace semantics and graded monads}}.
\newblock In: {\sl \bibinfo{booktitle}{LIPIcs-Leibniz International Proceedings
  in Informatics}}, \bibinfo{volume}{35}, \bibinfo{organization}{Schloss
  Dagstuhl-Leibniz-Zentrum fuer Informatik},
  \doi{10.4230/LIPIcs.CALCO.2015.253}.

\bibitemdeclare{article}{roberson2013variations}
\bibitem{roberson2013variations}
\bibinfo{author}{David~E \surnamestart Roberson\surnameend}:
  \emph{\bibinfo{title}{Variations on a theme: Graph homomorphisms. {P}h{D}'s
  thesis, {U}niversity of {W}aterloo, 2013.}}
\newblock \urlprefix\url{http://hdl.handle.net/10012/7814}.

\end{thebibliography}
\end{document}